\documentclass[letterpaper,11pt]{article}

\pdfoutput=1

\usepackage{amssymb}
\usepackage{amsmath}
\usepackage{amsthm}
\usepackage{algorithm}
\usepackage{graphicx}
\usepackage{caption}
\usepackage{multirow}
\usepackage{fullpage}
\usepackage{cmap}
\usepackage{bm}

\newtheorem{theorem}{Theorem}[section]
\newtheorem{corollary}[theorem]{Corollary}
\newtheorem{lemma}[theorem]{Lemma}

\newtheorem{claim}[theorem]{Claim}

\theoremstyle{definition}
\newtheorem{definition}[theorem]{Definition}
\theoremstyle{remark}
\newtheorem*{remark}{Remark}
\newtheorem{example}{Example}
\numberwithin{equation}{section}

\title{Do Capacity Constraints Constrain Coalitions?\thanks{An earlier version of this work has appeared in the Proceedings of the 29th AAAI Conference on Artificial Intelligence (AAAI-15). This work was partially supported by the European Research Council under the European Union's Seventh Framework Programme (FP7/2007-2013) / ERC grant agreement number 337122.}}

\author{Michal Feldman\\
Blavatnik School of Computer Science\\
Tel Aviv University
\and
Ofir Geri\\
Blavatnik School of Computer Science\\
Tel Aviv University}

\date{}

\begin{document}

\maketitle
\begin{abstract}
We study strong equilibria in symmetric capacitated cost-sharing connection games.
In these games, a graph with designated source $s$ and sink $t$ is given, and each edge is associated with some cost.
Each agent chooses strategically an $s$-$t$ path, knowing that the cost of each edge is shared equally between all agents using it.
Two settings of cost-sharing connection games have been previously studied: (i) games where {\em coalitions} can form, and (ii) games where edges are associated with {\em capacities}; both settings are inspired by real-life scenarios.
In this work we combine these scenarios and analyze strong equilibria (profiles where no coalition can deviate) in capacitated games.
This combination gives rise to new phenomena that do not occur in the previous settings.
Our contribution is two-fold.
First, we provide a topological characterization of networks that always admit a strong equilibrium.
Second, we establish tight bounds on the efficiency loss that may be incurred due to strategic behavior, as quantified by the strong price of anarchy (and stability) measures.
Interestingly, our results are qualitatively different than those obtained in the analysis of each scenario alone, and the combination of coalitions and capacities entails the introduction of more refined topology classes than previously studied.
\end{abstract}

\noindent \textbf{Keywords:} network congestion games, cost-sharing games, network design games, strong price of anarchy, strong equilibrium, coalitions, capacities, network topology

\section{Introduction}
The construction of networks by autonomous agents can be significantly affected by strategic behavior.
These situations are frequently modeled as {\em cost-sharing} connection games, which have been extensively studied \cite{Albers,FairCostSharing,ConnectionGame,CostSharingSE,capacitated}.
For example, consider the construction of a large computer network used by different countries, where each country is an autonomous system that serves its own strategic interests.

In cost-sharing games, a network is given and each edge is associated with a cost.
Each one of $n$ agents wishes to construct a path between its source and sink nodes, where the cost of each edge is shared equally between the agents who use it; each agent desires to minimize his individual cost.
Returning to our motivating example, a large computer network may be used by different countries, which should jointly cover its cost. All countries wish to use the network links, but prefer to do so at minimal cost.

The analysis of these games evolved around the Nash equilibrium (NE) notion --- a profile of strategies in which no agent can benefit from a unilateral deviation.
In recent years, two interesting cases of these games have been considered, both inspired by real-life scenarios.
First, a group of agents may form a {\em coalition} and collaborate for the benefit of all the members in the coalition.
This scenario is formalized using the notion of a {\em strong equilibrium} (SE) \cite{Aumann}.
A SE is a strategy profile from which no coalition can deviate in a way that benefits each one of its members.
Note that we are considering non-cooperative games with non-transferable utilities.
SE in cost-sharing connection games have been studied in \cite{CostSharingSE}.
The second scenario considers capacity constraints on the network edges \cite{capacitated}.
Here, each edge is associated with some capacity that limits the number of agents who can use it.
Both scenarios can naturally occur in our motivating example.
Indeed, countries may collaborate to improve their standing, and capacity constraints may arise due to bandwidth limits.

While each of these settings has been previously studied alone, the combination of the two has not been previously considered. This is the focus of the present paper.
Our main focus is on symmetric games, where all agents share the same source and sink nodes.
Interestingly, the combination of coalitions and capacities gives rise to new phenomena that do not occur in any of the previous scenarios.

\paragraph{Equilibrium existence.}
In the setting with no capacities or coalitions, the profile in which all agents use the lowest-cost path from the source to the sink is clearly a NE.
Now consider games with coalitions or capacities.
The profile just mentioned is clearly a SE as well, as no coalition can benefit by deviating from the lowest-cost path.
As for capacities, while capacity constraints may limit the use of the lowest-cost path, it has been proven that every (feasible) capacitated game admits a NE, due to the existence of an exact potential function \cite{capacitated}.
Consider next the combination of coalitions and capacities.
Here, even the mere existence problem becomes non-trivial.
Consider, for example, the game depicted in Figure \ref{fig:ext-para-ex}, played by two agents.
In our figures we use the notation $(x,y)$ to denote an edge with cost $x$ and capacity $y$.
In this game, the unique NE (up to renaming of agents) is one where one agent uses the upper edge ($a$) at cost $1$, and the other agent uses the path $b,c$ at cost $1.3$.
However, this profile is not a SE, as the two agents can deviate to the paths $b,c$ and $b,d$, respectively, at the respective costs of $0.7$ and $1.1$.
Since every SE is also a NE, and the only NE is not a SE, we conclude that this game does not admit a SE.

\begin{figure}
\centering
\includegraphics{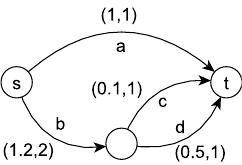}
\caption{\label{fig:ext-para-ex}A game with no SE}
\end{figure}

\paragraph{Efficiency loss.}
As another example, consider the efficiency loss incurred due to strategic behavior.
The standard measure used to quantify the efficiency loss is the price of anarchy (PoA), defined as the ratio between the cost of the worst-case NE and the optimal cost (where the social cost is defined as the overall cost of the edges in use).
In the game with no capacities or coalitions, it has been shown that the PoA can be as bad as $n$ (i.e., the number of agents), and this is tight \cite{ConnectionGame}.
The consideration of coalitions significantly reduces the loss; in fact, every SE is optimal.
In capacitated games, the PoA depends on the network topology.
For example, if the network topology adheres to a {\em series-parallel} structure (defined in Section \ref{sec:graph-theoretic}), then the PoA is bounded by $n$.
Do coalitions reduce the efficiency loss in capacitated games as in their uncapacitated counterparts?
Interestingly, the answer is {\em no} for series-parallel networks, but the answer is affirmative for a smaller set of topologies.
Once again, the combination of coalitions and capacities introduces interesting phenomena, and requires a more refined classification of network topologies than each setting alone.

\subsection{Our Results}

\paragraph{Equilibrium existence.}
As mentioned above, not all symmetric capacitated cost-sharing games admit a SE.
We provide a full characterization of network topologies that do admit a SE; i.e., every game played on a topology in this class admits a SE, and for every topology not in this class, there exists a game that does not possess a SE.
The analysis of this part requires the introduction of a new class of networks, which we refer to as {\em Series of Parallel Paths} (SPP) networks. This class is defined as all networks that are the concatenation of parallel-path networks.

\begin{table}
\centering
\begin{tabular}{|c|c|c|c|c|}
\cline{3-5}
\multicolumn{2}{c|}{} & \textbf{SPP and EP} & \textbf{SP} & \textbf{General} \\
\cline{1-5}
\hline
\multirow{2}{*}{\textbf{PoA}} & \textbf{Uncap.} & $n$ & $n$ & $n$ \\
\cline{2-5}
& \textbf{Cap.} & $n$ & $n$ & unbounded \\
\hline
\multirow{2}{*}{\textbf{SPoA}} & \textbf{Uncap.} & $1$ & $1$ & $1$ \\
\cline{2-5}
& \textbf{Cap. (*)} & $H_n$ & $n$ & unbounded \\
\hline
\textbf{SPoS} & \textbf{Cap. (*)} & $H_n$ & $\Theta(n)$ & unbounded \\
\hline
\end{tabular}
\caption{A comparison between upper bounds for the PoA in different scenarios. Cap. and uncap. are shorthands for capacitated and uncapacitated. Our results are marked with an asterisk (*). All bounds are tight.\label{table-summary}}
\end{table}

\paragraph{Efficiency loss.}
Previous analysis of the efficiency loss in cost-sharing games showed that the network topology significantly affects the incurred loss.
This observation is reinforced in this work, as summarized below (see also Table~\ref{table-summary} for a subset of our results).
We provide tight bounds on the strong price of anarchy (SPoA), defined as the ratio between the cost of the worst-case SE and the optimal cost and on the strong price of stability (SPoS), defined analogously with respect to the best-case SE.

Epstein et al.\ \cite{CostSharingSE} establish an upper bound of $H_n$ (i.e., the $n^{th}$ harmonic number, which is roughly $\log(n)$) on the SPoA in every game (including asymmetric games) that admits a SE.
As mentioned above, this result does not carry over to capacitated networks.
However, we show that this result does carry over to a specific class of network topologies, namely extension-parallel (EP) and SPP networks. Moreover, we provide an example showing that this bound is tight.
In series-parallel (SP) networks, the SPoA is at most $n$ (follows from the upper bound on the PoA), and we provide an example showing that this is tight.
For general networks, we show that the SPoA can be arbitrarily high, even in instances with only two agents.
Interestingly, our analysis results in a more refined classification of topologies than was required for the study of capacities or coalitions alone, most notably in the distinction between subclasses of SP networks.

In addition, we provide bounds on the SPoS.
We show that the SPoS can also be as high as $\Omega(n)$ in SP networks and is unbounded in general networks.
Interestingly, while the PoS is significantly better than the PoA ($1$ versus $n$ in uncapacitated games, and an even wider gap in capacitated games), in capacitated games we show that for all the topology classes we consider, the bounds on the SPoS and SPoA are asymptotically the same.
A natural interpretation of the PoS measure is the loss that is incurred if there exists a coordinator who can suggest an initial configuration to the agents.
While a coordinator can sometimes reduce the efficiency loss, our results here imply that a coordinator may not be useful in the worst case.

In addition to studying games with arbitrary capacity constraints, we consider the special case in which all the edges have the same capacity. We show that limiting the capacities to be homogeneous reduces the efficiency loss in games played on EP networks, in which every SE becomes a socially optimal solution.

\paragraph{Extensions.}
We consider two natural extensions to our model.
First, we consider asymmetric games. Here, we provide a characterization of network topologies that always admit a SE, and provide an example of a simple EP network for which the SPoA is unbounded.
Second, we show that all of our results regarding symmetric games extend to undirected networks as well.

\subsection{Related Work}
Since their introduction by Anshelevich et al.\ \cite{ConnectionGame}, cost-sharing connection games have been widely studied in various settings. In the original setting that was studied in \cite{ConnectionGame}, the agents' strategies consist of the amounts each of one of them is willing to pay for each edge, and in order to use an edge, the total amount that is paid for it should exceed its cost. This allows for general cost-sharing mechanisms that are not necessarily fair. The special case of fair cost-sharing was studied by Anshelevich et al.\ in \cite{FairCostSharing}. Under the fair cost-sharing rule, the strategy of an agent simply becomes a path in the graph, and the cost of each edge is shared equally between the agents who use it. They observed that the fair connection game belongs to the more general class of congestion games, which was introduced by Rosenthal \cite{CongestionGamesRosenthal}. In these games, a pure NE always exists, and they are known to admit an exact potential function, as defined in \cite{PotentialGames}.

The above works studied the inefficiency of equilibria by considering the measures of price of anarchy (PoA) and price of stability (PoS). The PoA measure was introduced by Koutsoupias and Papadimitriou \cite{PriceOfAnarchy} and has since been studied in numerous games and settings. In \cite{ConnectionGame}, it was shown that the PoA in every cost-sharing connection game is at most $n$ (the number of agents). It was also shown that the PoA can be as high as $n$ even in a simple network that consists of two parallel edges. For fair cost-sharing games in directed networks, a tight upper bound of $H_n$ on the PoS was established in \cite{FairCostSharing} using the potential function of the game. In the case of undirected networks, an upper bound of $O(\frac{\log{n}}{\log{\log{n}}})$ on the PoS in single-source games was established in \cite{LiMulticast}. In the special case of broadcast games, in which there is a single source and each other node is the sink node of a different agent, the PoS was shown to be a constant \cite{PoSUndirectedFOCS}. Lower bounds on the PoS in the undirected case were studied in \cite{PoSUndirectedTCS}, where it was shown that the PoS can be as high as $348/155$, $1.862$, and $20/11$ in general games, single-source games, and broadcast games, respectively.

The concept of a strong equilibrium --- a strategy profile that is resilient to coalitional deviations --- was introduced by Aumann \cite{Aumann}. Epstein et al.\ \cite{CostSharingSE} studied SE in fair and general cost-sharing connection games. Their results showed that the network topology plays a major role in the existence of a SE. For example, they showed that every single-source game that is played on a series-parallel graph has a SE, but this property does not hold for single-source games played on general graphs. The existence of SE in the more general class of network congestion games was studied in \cite{CongestionStrongPotential,NonDecCongestionSE}, and also in \cite{CongestionSE}, which established a characterization of the network topologies that admit SE in congestion games with monotone cost functions.

For fair cost-sharing, Epstein et al.\ \cite{CostSharingSE} and Albers \cite{Albers} independently proved an upper bound of $H_n$ on the strong price of anarchy (SPoA), which is the analogue of the PoA with respect to SE. This upper bound holds for every game that has a SE, regardless of the network topology. The SPoA measure was introduced in \cite{StrongPriceOfAnarchy}, and was also studied with respect to various additional settings; see, e.g., \cite{StrongPareto,LoadBalancing}.

The consideration of edge capacities in cost-sharing games was first suggested by Feldman and Ron \cite{capacitated}. They studied the PoA and PoS in symmetric games with regard to two objective functions: the sum of costs of all agents and the maximum cost of an agent. Their work has further emphasized the importance of the network topology. For example, when considering the sum of costs objective function, the PoA is at most $n$ in games played on series-parallel networks, but can be arbitrarily high in non-SP networks. A recent paper by Erlebach and Radoja \cite{furtherCapacitated} showed that the PoS with respect to the max-cost objective function is at most $n$ (in symmetric games), closing a gap that was presented in \cite{capacitated}.

\section{Model and Preliminaries}
\subsection{Symmetric Capacitated Fair Cost-Sharing Connection Games}
A symmetric capacitated fair cost-sharing connection (CFCSC) game is given by a tuple
\[
\Delta = (n,G=(V,E),s,t,\{p_e\}_{e \in E},\{c_e\}_{e \in E})
\]
where $n$ is the number of agents, $G=(V,E)$ is a directed graph, $s,t \in V$ are the source and sink nodes (respectively), and each edge $e \in E$ is associated with a cost $p_e \in \mathbb{R}^{\geq 0}$ and a capacity constraint $c_e \in \mathbb{N} \cup \{0\}$. Each agent wishes to construct an $s$-$t$ path in $G$ while maintaining minimal cost. The strategy space of agent $j$, denoted by $\Sigma_j$, is the set of all $s$-$t$ paths in $G$. The joint strategy space is denoted by $\Sigma=\Sigma_1 \times \ldots \times \Sigma_n$. The game is non-cooperative with non-transferable utilities.

Given a strategy profile $\bm{s}=(s_1,\ldots,s_n) \in \Sigma$, the number of agents that use an edge $e$ in the profile $\bm{s}$ is denoted by $x_{e}(\bm{s})=\left|\{j|e\in s_j\}\right|$. A profile $\bm{s}$ is said to be feasible if for every $e\in E$, $x_{e}(\bm{s})\leq c_{e}$. A game is said to be feasible if it admits a feasible strategy profile. Throughout this paper we only consider feasible games.
When all the edges have the same capacity, we say that the game has homogeneous capacities. Formally, in a game with homogeneous capacities, there is a number $c \in \mathbb{N}$ such that for every $e \in E$, $c_e = c$.

For a given strategy profile $\bm{s}$ and a coalition $C$, the induced strategy profile on the agents of the coalition $C$ is denoted by $\bm{s}_C$, and the strategy profile of the rest of the agents is denoted by $\bm{s}_{-C}$.
We consider the fair cost-sharing mechanism, where the cost of each edge is shared equally between all the agents who use it. The cost of agent $j$ in a strategy profile $\bm{s}$ is
\[
p_{j}(\bm{s})=\begin{cases}
\sum_{e\in s_{j}}\frac{p_{e}}{x_{e}(\bm{s})} & \text{if }\bm{s}\text{ is feasible}\\
\infty & \text{otherwise}
\end{cases}
\]

We use the utilitarian objective function, that is, the social cost of a strategy profile $\bm{s}$ is the sum of costs of all agents, $cost(\bm{s})=\sum_{j}{p_j(\bm{s})}$. The social cost of a profile is also equal to the sum of costs of the edges in use.

A strategy profile $\bm{s}$ is a Nash equilibrium (NE) if no agent can improve her cost by deviating to another strategy, i.e., for every $j$ and every strategy $s'_{j} \in \Sigma_{j}$, it holds that $p_{j}(\bm{s}) \leq p_{j}(s'_{j},\bm{s}_{-j})$ (where $\bm{s}_{-j}$ denotes the strategy profile of all agents except $j$ in $\bm{s}$). A strong equilibrium (SE) is a strategy profile in which no coalition can deviate jointly in a way that will strictly decrease the cost of every coalition member. Formally, a profile $\bm{s}$ is a SE if for every coalition of agents $C$ and every set of strategies $\bm{s}'_C \in \Sigma_{C}$, there exists an agent $j \in C$ such that $p_j(\bm{s}) \leq p_j(\bm{s}'_{C},\bm{s}_{-C})$.
The sets of NE and SE of a game $\Delta$ are denoted by $NE(\Delta)$ and $SE(\Delta)$, respectively.

We use the price of anarchy (PoA) and price of stability (PoS) measures to quantify the efficiency loss incurred due to strategic behavior. Let $\bm{s}^*$ be a strategy profile with minimal social cost in a game $\Delta$. Then, the PoA of $\Delta$ is the ratio between the cost of the worst-case NE and the cost of $\bm{s}^*$, namely $PoA=\max_{\bm{s} \in NE(\Delta)}{\frac{cost(\bm{s})}{cost(\bm{s}^*)}}$. Similarly, the PoS is the ratio between the cost of the best-case NE and the cost of $\bm{s}^*$, namely $PoS=\min_{\bm{s} \in NE(\Delta)}{\frac{cost(\bm{s})}{cost(\bm{s}^*)}}$. The analogues of the PoA and PoS with respect to SE are named the strong price of anarchy (SPoA) and strong price of stability (SPoS). Formally, $SPoA=\max_{\bm{s} \in SE(\Delta)}{\frac{cost(\bm{s})}{cost(\bm{s}^*)}}$ and $SPoS=\min_{\bm{s} \in SE(\Delta)}{\frac{cost(\bm{s})}{cost(\bm{s}^*)}}$.
For a family of games, these measures are defined with respect to the worst case over all the games in the family.

\subsection{Graph Theoretic Preliminaries\label{sec:graph-theoretic}}
A symmetric network is a graph $G=(V,E)$ with two designated nodes, a source $s\in V$ and a sink $t\in V$. Every node or edge in the network appears in at least one simple path from $s$ to $t$. The networks we consider in this paper are directed unless stated otherwise.
We hereby present three important operations on symmetric graphs.
\begin{itemize}
\item Identification: Given a graph $G=(V,E)$, the identification of two nodes $v_1,v_2 \in V$ yields a new graph $G'=(V',E')$, where $V'=(V \cup \{v\}) \backslash \{v_1,v_2\}$ and $E'$ includes all the edges of $E$, where each edge that was connected to $v_{1}$ or $v_{2}$ is now connected to $v$ instead. Figuratively, the identification operation is the collapse of two nodes into one.
\item Series composition: Given two symmetric networks, $G_1=(V_1,E_1)$ with $s_1,t_1 \in V_1$ and $G_2=(V_2,E_2)$ with $s_2,t_2 \in V_2$, the series composition $G=G_1 \rightarrow G_2$ is the network formed by identifying $t_1$ and $s_2$ in the union network $G'=(V_1 \cup V_2,E_1 \cup E_2)$. In the composed network $G$, the new source is $s_1$ and the new sink is $t_2$.
\item Parallel composition: Given two symmetric networks, $G_1=(V_1,E_1)$ with $s_1,t_1 \in V_1$ and $G_2=(V_2,E_2)$ with $s_2,t_2 \in V_2$, the parallel composition $G=G_1 \parallel G_2$ is the network formed by identifying the nodes $s_1$ and $s_2$ (forming a new source $s$) and the nodes $t_1$ and $t_2$ (forming a new sink $t$) in the union network $G'=(V_1 \cup V_2,E_1 \cup E_2)$.
\end{itemize}

Using these operations, we define three classes of network topologies that will be of interest throughout the paper.

\begin{definition}
A symmetric network $G=(V,E)$ is a {\em series-parallel} (SP) network if it consists of a single edge, or if there are two SP networks $G_1,G_2$ such that $G = G_1 \rightarrow G_2$ or $G = G_1 \parallel G_2$.
\end{definition}

\begin{definition}
A symmetric network $G=(V,E)$ is an {\em extension-parallel} (EP) network if one of the following applies:
\begin{enumerate}
\item $G$ consists of a single edge.
\item There are two EP networks $G_1,G_2$ such that $G = G_1 \parallel G_2$.
\item There is an EP network $G_1$ and an edge $e$ such that $G = G_1 \rightarrow e$ or $G = e \rightarrow G_1$.
\end{enumerate}
\end{definition}

\begin{definition}
A symmetric network $G$ is a {\em Series of Parallel Paths} (SPP) if there exist networks $G_1, \ldots, G_k$, each constructed by a parallel composition of simple paths, such that $G=G_1 \rightarrow \ldots \rightarrow G_k$.
\end{definition}

Note that every EP network is an SP network. An SPP network may or may not be an EP network, but it is always an SP network. Examples of SPP, EP, SP, and non-SP networks appear in Figures \ref{fig:SPP_No_SE}, \ref{fig:ext-para-ex}, \ref{fig:SP-SPoAn}, and \ref{fig:braess-ex}, respectively.

Finally, we have to define when a network is embedded in another network. A symmetric network $G$ is embedded in a network $G'$ if $G'$ is isomorphic to $G$ or to a network derived from $G$ using any number of the following operations:
\begin{itemize}
\item Subdivision: replacing an edge $(u,v)$ by a new node $w$ and two edges $(u,w)$ and $(w,v)$
\item Addition: adding a new edge connecting two existing nodes (including nodes that were added using subdivision or extension)
\item Extension: adding a new source or sink node and an edge connecting the new node with the original source or sink node, respectively
\end{itemize}

\section{Existence of Strong Equilibria\label{sec:existence}}
Every symmetric cost-sharing game admits a SE, as all agents can share the lowest-cost $s$-$t$ path \cite{CostSharingSE}.
In the capacitated version, it has been shown by Feldman and Ron \cite{capacitated} that a pure NE exists in every feasible game, by establishing that the game admits a potential function.
Therefore, the consideration of capacities or coalitions alone does not preclude the existence of an equilibrium.
However, as was already observed in the introduction (recall Figure \ref{fig:ext-para-ex}), capacitated games may not admit any SE, even when played on a simple EP network (as shown in the example).
We present an additional example of a game that does not admit a SE, which is played on an underlying network known as a Braess graph.

\begin{figure}
\centering
\includegraphics{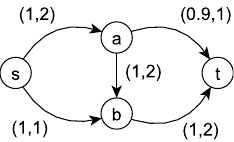}
\caption{\label{fig:braess-ex}A game with no SE played on a Braess graph}
\end{figure}

\begin{example}
Consider the game of two agents that is depicted in Figure \ref{fig:braess-ex}. We show that each of the possible strategy profiles is not a SE.
If both agents use the path $(s,a,b,t)$, one of the agents can profit by deviating to $(s,a,t)$.
If one agents uses $(s,a,t)$ and the other uses $(s,a,b,t)$, it is profitable for the latter to deviate to $(s,b,t)$.
If one agent uses $(s,a,b,t)$ and the other uses $(s,b,t)$, the former can reduce her cost by using $(s,a,t)$ instead.
A strategy profile in which one agent uses $(s,a,t)$ and the other uses $(s,b,t)$ is a pure NE, but it is not a SE because both agents can reduce their costs if they both use the path $(s,a,b,t)$. We conclude that this game does not admit a SE.
\end{example}

As it turns out, the two networks depicted in Figures \ref{fig:ext-para-ex} and \ref{fig:braess-ex} are, roughly speaking, the only barriers to SE existence. This is formalized in the remainder of this section as an exact characterization of the network topologies that always admit a SE.

\begin{definition}
A symmetric network $G$ is said to admit a SE if {\em every} symmetric CFCSC game played on $G$ admits a SE.
\end{definition}

Recall that a symmetric network is defined as a graph with designated source and sink nodes.
According to the last definition, a network $G$ is said to admit a SE if {\em every} CFCSC game played on $G$ (i.e., for every assignment of costs and capacities to the edges, and for every number of agents) admits a SE.
Conversely, a network $G$ does not admit a SE if there exists an example of a game played on $G$ that does not admit a SE.
The following theorem establishes the characterization of networks that admit a SE.

\begin{theorem}\label{existence-iff-thm}
A symmetric network $G$ admits a SE if and only if $G$ is an SPP network.
\end{theorem}

The proof is divided into two parts. In Theorem \ref{thm:spp-se-existence} we prove that every SPP network admits a SE, and in Theorem \ref{thm:no-se-on-non-spp} we establish the converse direction.

\begin{theorem}\label{thm:spp-se-existence}
Every SPP network admits a SE.
\end{theorem}

\begin{algorithm}[tb]
\caption{Compute a SE for a network of parallel edges\label{compute-se-alg}}
\begin{enumerate}
\item $j \leftarrow 1$
\item While there are agents that have not been assigned edges:
\begin{enumerate}
\item For every edge $e$ that has not been assigned to agents yet, compute its {\em fractional cost} $\frac {p_e} {\min\{c_e,n-\sum_{i=1}^{j-1}{n_i}\}}$.
\item Find the edge $e$ with the minimal fractional cost, and compute $n_j=\min\{c_e,n-\sum_{i=1}^{j-1}{n_i}\}$.
\item Assign $e$ to the following $n_j$ agents: $\sum_{i=1}^{j-1}{n_i}+1,\ldots,\sum_{i=1}^{j}{n_i}$.
\item Increment $j$.
\end{enumerate}
\end{enumerate}
\end{algorithm}

\begin{proof}
We first observe that in our context, networks of parallel paths can be reduced to networks of parallel edges, where each path is replaced by an edge with capacity that equals the minimal capacity on the path and cost that equals the total cost of the path.
Therefore, it suffices to prove the assertion of the theorem for networks constructed by series composition of parallel-edge graphs.

Our main building block in this proof is Algorithm \ref{compute-se-alg}, which computes a SE for games played on parallel-edge networks.
The algorithm assigns edges to agents in a greedy manner. It first finds the edge with the lowest \emph{fractional cost} (as defined in the description of the algorithm). When used to its maximum capacity (or by all the agents), this edge has the lowest possible cost for a single agent. The algorithm then assigns this edge to as many agents as possible. In the description of the algorithm, the number of agents who use the first edge is denoted by $n_1$.
In each iteration, the algorithm proceeds in the same way: it finds the edge with the lowest fractional cost that is still available and assigns it to as many agents as possible (this number is limited by the capacity of the edge and the number of remaining agents). The number of agents who use the $j^{th}$ edge is denoted by $n_j$. The algorithm terminates within at most $n$ steps.

As part of the proof, we show that the returned profile is a SE. An important property of this algorithm is the \emph{entireness property}: Each edge that is used by the algorithm is used to its maximum capacity, except for the last edge that is chosen by the algorithm, which may or may not be entirely used depending on the number of remaining agents. Moreover, the cost incurred due to an edge is exactly its fractional cost in the iteration in which it was chosen, and the costs of the agents are non-decreasing in the agents' indices.

After defining Algorithm \ref{compute-se-alg}, we are ready to prove the theorem.
Let $G_1, \ldots, G_k$ be parallel-edge networks, and let $G$ be a series composition of these networks. We define a strategy profile $\bm{s}$ for the game played on $G$ in the following way: For every $i=1,\ldots,k$, we compute a profile $\bm{s}^i$ for network $G_i$ using Algorithm \ref{compute-se-alg}. $\bm{s}$ is defined to be the strategy profile in which the strategy of agent $j$ is the concatenation of his strategies in $\bm{s}^1,\ldots,\bm{s}^k$. We prove that $\bm{s}$ is a SE.

From the definition of Algorithm \ref{compute-se-alg}, we get that for every profile $\bm{s}^i$, and for every two agents $j,j'$ such that $j<j'$, $p_j(\bm{s}^i) \leq p_{j'}(\bm{s}^i)$. In other words, the first agent in $\bm{s}$ incurs the lowest cost in each of $G_1,\ldots,G_k$, the next agent incurs the second lowest cost in each of the subnetworks, and so on.

The intuition behind the proof is the following: Since the first agent pays the lowest possible cost in each of the subnetworks, she cannot reduce her cost in any subnetwork and thus cannot participate in any profitable coalitional deviation. Assuming that the first agent does not participate in the deviation, the second agent cannot reduce her cost in any subnetwork and cannot participate in the deviation as well. Arguments of this nature will be used to show that none of the agents can participate in the coalitional deviation.

Formally, assume by contradiction that there is a coalition $C$ for which there is a profitable deviation, yielding a strategy profile $\bm{s}'$. Let $j$ be the minimal index of an agent in $C$, i.e., $j$ is the agent that uses the $j^{th}$ lowest-cost edge in each of $\bm{s}^1,\ldots,\bm{s}^k$.

In order for $\bm{s}'$ to be a profitable coalitional deviation, $j$ must reduce her cost in at least one of $G_1,\ldots,G_k$. Assume, without loss of generality, that after the coalition deviates to $\bm{s}'$, agent $j$ reduces her cost in $G_1$. Denote the edge used by $j$ in $\bm{s}^1$ by $e$. It follows that there exists an edge $e'$ in $G_1$ such that
\[
\frac{p_{e'}}{x_{e'}(\bm{s}')} < \frac{p_e}{x_e(\bm{s})}.
\]

First, we show that $e' \ne e$. From the last inequality, if $e' = e$, the number of agents who use $e$ in $\bm{s}'$ must be larger than the number of agents who use $e$ in $\bm{s}$. We consider each of the edges that are used in $\bm{s}^1$ and show that the number of agents who use it cannot increase after the deviation. If $e$ is used to its maximum capacity, the number of agents who use it cannot increase. By the entireness property of Algorithm \ref{compute-se-alg}, we know that the only edge in $\bm{s}^1$ that may not be entirely used is the last edge that was picked by the algorithm. So if the edge $e$ that agent $j$ uses in $\bm{s}^1$ is the last edge chosen by the algorithm, all the agents with higher indices must also use $e$ in $\bm{s}^1$. Particularly, since $j$ is the agent with the minimal index in $C$, it follows that all the agents in the coalition use $e$ in $\bm{s}^1$. Therefore, the number of agents who use $e$ cannot increase after the deviation in this case as well. We have shown that the number of agents who use $e$ in $\bm{s}'$ cannot be larger than the number of agents who use $e$ in $\bm{s}$, and conclude that $e' \ne e$.

Second, we show that $e'$ cannot be one of the edges that are entirely used in $\bm{s}^1$. Since $e' \ne e$, $e'$ was picked by the algorithm either before or after the iteration in which $j$ was assigned an edge. Algorithm \ref{compute-se-alg} picks edges in non-decreasing order of fractional cost, so agent $j$ can only reduce her cost if $e'$ was picked by the algorithm before $j$ was assigned an edge. In that case, $e'$ is only used by agents with indices lower than $j$. However, these agents do not participate in the deviation ($j$ is the agent with the minimal index in $C$), and $j$ cannot deviate to $e'$, which is used to its maximum capacity. Therefore, $e'$ cannot be one of the edges that are entirely used in $\bm{s}^1$.

We are left with two options for edge $e'$:
\begin{enumerate}
\item $e'$ is partially used in $\bm{s}^1$ (i.e., $e'$ is the last edge that was chosen by Algorithm \ref{compute-se-alg}). Since $e' \ne e$, the edge that agent $j$ uses in $\bm{s}^1$ was picked by the algorithm before $e'$. So $e'$ is used in $\bm{s}^1$ only by agents indexed $j+1$ or higher. The only agents who may use $e'$ in $\bm{s}'$ are agents who use $e'$ in $\bm{s}^1$ or agents that belong to $C$. These two sets of agents only contain agents with indices $j$ or higher, thus $x_{e'}(\bm{s}') \leq \left|\{j,\ldots,n\}\right| = n-j+1$.
\item $e'$ is not used in $\bm{s}^1$. In this case, the agents who use $e'$ in $\bm{s}'$ must be members of the coalition $C \subseteq \{j,\ldots,n\}$. Therefore, $x_{e'}(\bm{s}') \leq \left|\{j,\ldots,n\}\right| = n-j+1$.
\end{enumerate}

In both cases, we have shown that $x_{e'}(\bm{s}') \leq n-j+1$, and since the number of agents who use an edge is at most its capacity, $x_{e'}(\bm{s}') \leq \min\{c_{e'},n-j+1\}$. From the assumption that agent $j$ reduces her cost in $G_1$, we get that
\[
\frac{p_{e'}}{\min\{c_{e'},n-j+1\}} \leq \frac{p_{e'}}{x_{e'}(\bm{s}')} < \frac{p_e}{x_e(\bm{s})}
\]
and the algorithm should have assigned $e'$ to agent $j$ (or a previous agent) instead of $e$. We get a contradiction and conclude that $\bm{s}$ is a SE.
\end{proof}

We now prove the converse direction, namely that for every non-SPP network, there is a CFCSC game played on it that does not admit a SE.
The key lemma in our proof connects us back to the two examples mentioned in the beginning of this section.

\begin{lemma}\label{lem:spp-not-embedding}
A symmetric network is SPP if and only if it does not embed any of the networks depicted in Figure \ref{fig:non-spp}.
\end{lemma}
\begin{proof}
We first show that every SPP network does not embed any of the networks of Figure \ref{fig:non-spp}. Note that in a network of parallel paths, no node except for the source and the sink can have two edges entering or leaving it. Thus, a network of parallel paths cannot embed any of the networks of Figure \ref{fig:non-spp}. For an SPP network to embed one of the forbidden networks (the networks that appear in the figure), at least one of its subnetworks of parallel paths has to embed a forbidden network, which is impossible. We conclude that if a network is SPP, it cannot embed any of these networks.

In the converse direction, let $G$ be a network that does not embed any of the networks of Figure \ref{fig:non-spp}. We prove that $G$ is SPP. It is proven in \cite{CongestionSE} that a network is SP if and only if it does not embed a Braess graph (Figure \ref{fig:non-spp}(a)). Therefore, $G$ is SP. The rest of the proof is by induction on the height of the construction tree of $G$. The base case is a single edge, which is an SPP network. The induction hypothesis states that every SP network that does not embed any of the forbidden networks and has a construction tree with height smaller than that of $G$ is SPP.

Assume, by contradiction, that $G$ is not SPP. Since $G$ is SP, we know that either $G=G_1 \rightarrow G_2$ or $G=G_1 \parallel G_2$. By induction, we can assume that both $G_1$ and $G_2$ are SPP (otherwise, one of $G_1$ and $G_2$ is an SP network with smaller construction tree, that does not embed any of the forbidden networks and is not an SPP network). The series composition of two SPP networks is an SPP network, so it is impossible that $G=G_1 \rightarrow G_2$. Hence, $G=G_1 \parallel G_2$. If both of $G_1,G_2$ are networks of parallel paths, their parallel composition is a network of parallel paths, which is SPP. Therefore, one of $G_1,G_2$ embeds a series composition of a single edge and a parallel-edge network. It follows that $G$ embeds one of the networks presented in Figure \ref{fig:non-spp}(b)-(c). It is a contradiction, hence $G$ is SPP.
\end{proof}

\begin{figure}
\centering
\includegraphics{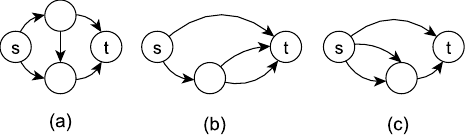}
\caption{\label{fig:non-spp}Minimal non-SPP networks}
\end{figure}

Using Lemma~\ref{lem:spp-not-embedding} we can now prove the following. The proof is similar to that of a lemma proven in \cite{capacitated}, and is specified here for completeness.

\begin{theorem}\label{thm:no-se-on-non-spp}
For every non-SPP network $G$, there exists a symmetric CFCSC game played on $G$ that does not admit a SE.
\end{theorem}
\begin{proof}
By Lemma~\ref{lem:spp-not-embedding}, if $G$ is not an SPP network, it must embed one of the networks depicted in Figure \ref{fig:non-spp}. We define a game of two agents played on $G$. We start from one of the networks presented in Figure \ref{fig:non-spp} and assign costs and capacities as in Figures \ref{fig:ext-para-ex} and \ref{fig:braess-ex}. The network $G$ is derived from the forbidden network using the operations: subdivision, addition, and extension. We show how to assign costs and capacities to the edges that are added in the process such that each strategy profile in the new game corresponds to a strategy profile in the original game, which does not admit a SE.

Formally, let $\Delta$ be a CFCSC game played on a network $G$, and let $G'$ be a network derived from $G$ by applying one of the operations: subdivision, addition, or extension. We construct a game $\Delta'$ played on $G'$ that emulates $\Delta$: We define a bijection $f$ from the strategy space of an agent in $\Delta$ to the strategy space of an agent in $\Delta'$, such that $(s_1,\ldots,s_n)$ is feasible whenever $(f(s_1),\ldots,f(s_n))$ is feasible, and $p_j(s_1,\ldots,s_n) = p_j(f(s_1),\ldots,f(s_n))$ for every strategy profile $(s_1,\ldots,s_n)$. We do not need to define different mappings between strategy spaces for different agents, as we consider symmetric games and the strategy spaces of the agents are identical. We define $\Delta'$ and the function $f$ for each one of the possible operations:

\begin{itemize}
\item Subdivision: If an edge $e$ is subdivided into $e_1$ and $e_2$, we set $c_{e_1}=c_{e_2}=c_e$, $p_{e_1}=p_e$, and $p_{e_2}=0$. Given a strategy $s_j$ in $\Delta$, $f(s_j)$ is the same path as $s_j$ where the edge $e$ is replaced by the path $(e_1, e_2)$ (if $e \notin s_j$, $f(s_j) = s_j$). No strategy profile in $\Delta'$ can use only one of $e_1, e_2$, so $f$ is a bijection.
\item Addition: If an edge $e$ connecting two existing nodes is added, we set its capacity to be $c_e=0$, and define $f(s_j) = s_j$ for every strategy $s_j$. No strategy in $\Delta'$ can use the new edge, so $f$ is a bijection.
\item Extension: If the source or the sink are extended using a new edge $e$, we set $c_e=n$ and $p_e=0$. Given a strategy $s_j$, $f(s_j)$ uses the newly added edge $e$ concatenated to the path used in $s_j$. Every strategy in $\Delta'$ must use the new edge $e$, concatenated to a path in $\Delta$, so $f$ is a bijection.
\end{itemize}

We claim that $\Delta'$ does not admit a SE if and only if $\Delta$ does not admit a SE. Let $\bm{s},\bm{s}'$ be strategy profiles in $\Delta$, and let $C$ be a coalition. The profile $(s'_1,\ldots,s'_n)$ is a profitable deviation from $(s_1,\ldots,s_n)$ for $C$ in $\Delta$, if and only if the profile $(f(s'_1),\ldots,f(s'_n))$ is a profitable deviation from $(f(s_1),\ldots,f(s_n))$ for $C$ in $\Delta'$.
\end{proof}

This concludes the proof of Theorem \ref{existence-iff-thm}.

\section{Strong Price of Anarchy}
\subsection{EP and SPP Networks}
In this section we bound the strong price of anarchy (SPoA) in capacitated games that admit SE.
The following theorem establishes an upper bound on the SPoA for EP networks.
\begin{theorem}
\label{thm:ep-spoa}
For every symmetric CFCSC game played on an EP network, it holds that $SPoA \leq H_n$ (if a SE exists), and this bound is tight.
\end{theorem}
\begin{proof}
We first prove the upper bound. Epstein et al.\ \cite{CostSharingSE} showed that $SPoA \leq H_n$ for uncapacitated cost-sharing games.
In their proof, they used the fact that no coalition can beneficially deviate from a strategy profile $\bm{s}$ that is a SE (by definition).
In particular, if some coalition $C$ deviates to its corresponding profile in the socially optimal profile $\bm{s}^*$, then one of the agents in $C$ weakly prefers the initial profile $\bm{s}$ to the new profile $(\bm{s}^*_{C},\bm{s}_{-C})$.
The desired bound is then derived by the obtained inequalities for coalitions of sizes $n, \ldots, 1$.
The only barrier to applying the exact same technique to capacitated games is the fact that the deviation into profile $(\bm{s}^*_{C},\bm{s}_{-C})$ might be infeasible due to capacity constraints.
Our key lemma in this section shows that for games played on EP networks there always exists such a feasible deviation. The lemma is sufficient in order to prove the upper bound on the SPoA. Note that there may be more than a single socially optimal profile, but we have the freedom to choose one of them to be $\bm{s}^*$.

\begin{lemma}
\label{feasible-combined-profile-ep}
Let $G$ be an EP network and $\bm{s}$ be a SE in a symmetric CFCSC game played on $G$. There exists a feasible strategy profile $\bm{s}^*$ such that the cost of $\bm{s}^*$ is minimal, and for every coalition $C$, the profile $(\bm{s}_C,\bm{s}^*_{-C})$ is feasible.
\end{lemma}

The following lemma, which is due to Feldman and Ron \cite{capacitated}, will be used in the proof.

\begin{lemma}
\label{complete-partial-profile-lemma}
\cite{capacitated} Let $G$ be an SP network. Let $\bm{s}$ be a feasible profile of $k$ agents in a game played on $G$, and let $\bm{s}'$ be a feasible profile of $r$ agents such that $r < k$. There exists an $s-t$ path in $G$ that is feasible together with the strategies of the $r$ agents in $\bm{s}'$ and uses only edges that are used in the profile $\bm{s}$.
\end{lemma}

\begin{proof}[Proof of Lemma \ref{feasible-combined-profile-ep}]
Let $G_{OPT}$ be the subnetwork that contains only the edges that are used by an optimal strategy profile, and let $N$ denote the set of agents. We first define a specific profile $\bm{s}^*$ played on $G_{OPT}$, and then prove that for every coalition $C$, the strategy profile $(\bm{s}_C,\bm{s}^*_{-C})$ is feasible. Since $G$ is EP, $G_{OPT}$ is also EP.\footnote{This can observed by constructing $G_{OPT}$ using the same construction tree of $G$, excluding parallel composition operations of subnetworks that all their edges are not used in $G_{OPT}$.} We define $\bm{s}^*$ in two steps: First, we assign a strategy to as many agents as possible using recursion (on the structure of $G_{OPT}$). Then, we use Lemma \ref{complete-partial-profile-lemma} to extend this set of strategies to a profile of all agents. The profile $\bm{s}^*$ is defined using Algorithm \ref{alg-find-optimal-profile}, which chooses a specific profile from all the optimal strategy profiles. The algorithm gets as input the subnetwork $G_{OPT}$ that is used by an optimal profile and the strategy profile $\bm{s}$, which is a SE.

It is important to note that the algorithm might possibly define a strategy profile only for a subset of the agents.
In step 2, $N_1 \cap N_2 = \phi$, but it is possible that $N_1 \cup N_2 \subset N$ (since $G_{OPT}$ is a subnetwork of $G$). Therefore, by the end of step 2, it is possible that not all the agents in $N$ are assigned a strategy. In step 3(c), it is also possible that not all the agents in $N$ will be assigned a strategy (as an agent is assigned a strategy only if there is a path available in $G_1$). Thus, we have to extend this strategy profile to a feasible profile for all agents (using only edges from $G_{OPT}$). We do so by applying Lemma \ref{complete-partial-profile-lemma} to the network $G_{OPT}$, where $\bm{s}$ is any strategy profile that uses only edges from $G_{OPT}$, and $\bm{s}'$ is the partial profile that was defined above. This provides a full definition of the profile $\bm{s}^*$. Example~\ref{ex:alg-find-optimal-profile} that follows this proof illustrates the definition of $\bm{s}^*$.

\begin{algorithm}[tb]
\caption{Choosing the optimal profile $\bm{s}^*$}
\label{alg-find-optimal-profile}
\emph{Input:} $G_{OPT}$ is a graph, $N$ is a set of agents, and $\bm{s}$ is a strategy profile.

\emph{ChooseOptimalProfile}($G_{OPT}$, $N$, $\bm{s}$):
\begin{enumerate}
\item If $G_{OPT} = e$, where $e$ is a single edge, return a strategy profile in which the edge $e$ is assigned to all the agents in $N$ who use $e$ in $\bm{s}$.
\item If $G_{OPT} = G_1 \parallel G_2$:
 \begin{enumerate}
 \item Let $N_i \subseteq N$ be the set of agents that use an edge of $G_i$ ($i=1,2$).
 \item $\bm{s}^1 \leftarrow ChooseOptimalProfile(G_1,N_1,\bm{s})$
 \item $\bm{s}^2 \leftarrow ChooseOptimalProfile(G_2,N_2,\bm{s})$
 \item Return the union of the profiles: $(\bm{s}^1,\bm{s}^2)$.
 \end{enumerate}
\item If $G_{OPT} = G_1 \rightarrow e$ or $G_{OPT} = e \rightarrow G_1$, where $e$ is an edge:
 \begin{enumerate}
 \item $\bm{s}^1 \leftarrow ChooseOptimalProfile(G_1,N,\bm{s})$
 \item Each agent that has a strategy in $\bm{s}^1$ will also use the edge $e$.
 \item For any other agent that uses $e$ in $\bm{s}$, attempt to find an available path in $G_1$. If found, assign it together with the edge $e$ to the agent.
 \item Return the profile that was defined in the last three steps.
 \item[$\bullet$] In case it is possible to represent $G_{OPT}$ as both $G_1 \rightarrow e_1$ and $e_2 \rightarrow G_2$, choose the representation in which the edge $e_i$ is used by the maximal set of agents in $\bm{s}$.
 \end{enumerate}
\end{enumerate}
\end{algorithm}

We claim that the profile $\bm{s}^*$ satisfies the capacity constraints. In step 1 of the algorithm, an edge is assigned to the same agents that use it in the feasible profile $\bm{s}$ (due to the way $N$ is split in step 2). In step 3, the edge $e$ is used only by agents that use it in $\bm{s}$, and edges in $G_1$ are assigned in a way that satisfies the capacity constraints. Therefore, no edge exceeds its capacity.

Let $C$ be a coalition of agents. We prove that the strategy profile $\bm{s}_{comb} = (\bm{s}_C,\bm{s}^*_{-C})$ is feasible. Let $e$ be an edge. Let $M$ denote the set of agents that use $e$ in $\bm{s}$, and let $M^*$ denote the set of agents that use $e$ in $\bm{s}^*$. The set of agents that use $e$ in $\bm{s}_{comb}$ is $(M \cap C) \cup (M^* \cap (N \backslash C))$. If $e$ is used only in $\bm{s}$, it cannot exceed its capacity in $\bm{s}_{comb}$, since $M \cap C \subseteq M$, and the profile $\bm{s}$ is feasible. The same applies to edges that are used only in $\bm{s}^*$. It remains to consider edges that are used in both $\bm{s}$ and $\bm{s}^*$. There are two types of edges in EP networks: ones that are added to the graph through an extension of the source or sink, and all other edges. The algorithm assigns edges of the former type to agents in step 3, and edges of the latter type in step 1.

If edge $e$ is used in both $\bm{s}$ and $\bm{s}^*$, there are two cases:
\begin{enumerate}
\item If $e$ is assigned to agents in step 1 of the algorithm, then the set $N$ in this step contains all the agents that use $e$ in $\bm{s}$, namely, $M \subseteq N$, thus, $M \subseteq M^*$ ($e \in s_j \Rightarrow e \in s^*_j$). The set of agents that use $e$ in $\bm{s}_{comb}$ is $(M \cap C) \cup (M^* \cap (N \backslash C)) \subseteq M^*$. Since $\bm{s}^*$ is feasible, then $e$ does not exceed its capacity in $\bm{s}_{comb}$.
\item If $e$ is added to the network $G_{OPT}$ using extension of the source or the sink, the algorithm assigns $e$ to agents in step 3. There are two sub-cases.
\begin{enumerate}
\item In step 3(c), each agent that uses $e$ in $\bm{s}$ was assigned a path in $G_1$ in $\bm{s}^*$. In that case, $M \subseteq M^*$. The set of agents that use $e$ in $\bm{s}_{comb}$ is $(M \cap C) \cup (M^* \cap (N \backslash C)) \subseteq M^*$. Since $\bm{s}^*$ is feasible, we get that $e$ does not exceed its capacity in $\bm{s}_{comb}$.
\item There is an agent that uses $e$ in $\bm{s}$ but does not use $e$ and $G_1$ in $\bm{s}^*$. Due to the definition of $\bm{s}^*$, there is no available path in $G_1$. Therefore, no agent can be assigned a path in $G_1$ later. Each agent that uses $e$ in $\bm{s}^*$ must use $G_1$. Hence, no agents will be assigned the edge $e$ after it is used in step 3 of the algorithm. Thus, $M^* \subseteq M$. The set of agents that use $e$ in $\bm{s}_{comb}$ is $(M \cap C) \cup (M^* \cap (N \backslash C)) \subseteq M$. Since $\bm{s}$ is feasible, then $e$ does not exceed its capacity in $\bm{s}_{comb}$.
\end{enumerate}
\end{enumerate}

We conclude that no edge exceeds its capacity in $\bm{s}_{comb}$.
\end{proof}

\begin{figure}
\centering
\includegraphics{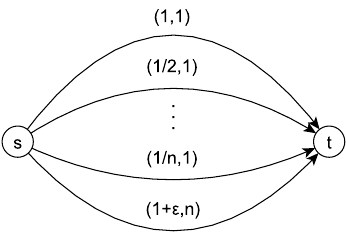}
\caption{\label{fig:PoS-Hn-tightness}\cite{capacitated} A game played on an EP network with $PoS=\frac{H_n}{1+\epsilon}$}
\end{figure}

We turn to discuss the tightness of the upper bound established by the previous lemma. Feldman and Ron \cite{capacitated} show an example of a game with $n$ agents, in which the PoS is $\frac{H_n}{1+\epsilon}$ for every $\epsilon > 0$. Their example is included here as Figure~\ref{fig:PoS-Hn-tightness}. The only NE in this example is also a SE. This shows that the upper bound of $H_n$ on the SPoA in EP networks is tight.
\end{proof}

\begin{example}\label{ex:alg-find-optimal-profile}
\begin{figure}
\centering
\includegraphics{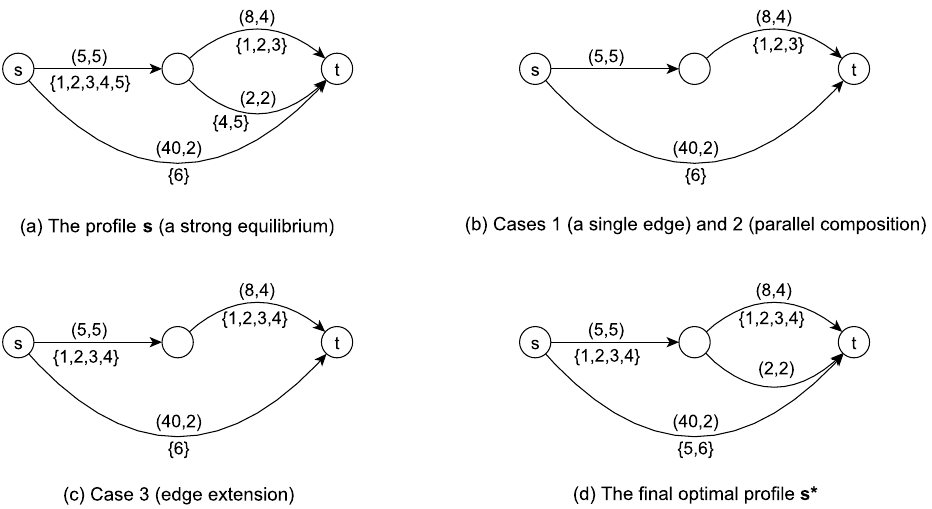}
\caption{\label{fig:FindOptimalProfileExample}An illustration of Algorithm \ref{alg-find-optimal-profile}}
\end{figure}

Figure \ref{fig:FindOptimalProfileExample} provides an example that illustrates the definition of $\bm{s}^*$.
Recall that Algorithm~\ref{alg-find-optimal-profile} gets as input a SE ($\bm{s}$) and the subnetwork that is used by an optimal strategy profile ($G_{OPT}$). The algorithm is used to choose an optimal profile $\bm{s}^*$ from all possible optimal strategy profiles that use only edges from $G_{OPT}$.
Figure \ref{fig:FindOptimalProfileExample}(a) shows the strategies of $n=6$ agents in the profile $\bm{s}$, which is a SE. Figures \ref{fig:FindOptimalProfileExample}(b)-(d) show the construction of $\bm{s}^*$ over $G_{OPT}$. Figure \ref{fig:FindOptimalProfileExample}(b) illustrates the case of parallel composition: The lower edge is used by agent $6$ (as in profile $\bm{s}$), and we use the algorithm recursively in order to find a profile for agents $1,\ldots,5$ who use the upper subnetwork. Figure \ref{fig:FindOptimalProfileExample}(c) shows the case of edge extension: The edge of cost 8 is a subnetwork that was extended with the edge of cost 5. The edge of cost $8$ is used by the agents $1,2,3$ in $\bm{s}$, hence it will be used by the same agents in $\bm{s}^*$. The algorithm now defines which agents will use the edge of cost 5. First of all, it has to be used by agents $1,2,3$. We consider the rest of the agents that use this edge in $\bm{s}$. One of them (agent $4$) can use the edge of cost 8, and the other (agent $5$) can not. Therefore, agent $4$ is assigned the same path as agents $1,2,3$, and the strategy of agent $5$ is left undefined. Figure \ref{fig:FindOptimalProfileExample}(d) shows the final profile $\bm{s}^*$, after it was extended to all agents using Lemma \ref{complete-partial-profile-lemma}.
\end{example}

\begin{remark}
Let $c_{max} = \max_{e \in E} c_e$ be the maximal capacity of an edge. In the case of capacitated games, the bound of $H_n$ on the SPoA that was established in \cite{CostSharingSE} can be slightly improved to $H_{c_{max}}$. In their proof, Epstein et al.\ used the potential function to bound the cost of every SE. For a SE $\bm{s}$ and an optimal solution $\bm{s}^*$, they have shown that $cost(\bm{s}) \leq \Phi(\bm{s}^*)$, where the potential function $\Phi$ is defined by $\Phi(\bm{s}) = \sum_{e \in E}{p_e \cdot H_{x_e(\bm{s})}}$ (where $x_e(\bm{s})$ is the number of agents who use edge $e$ in $\bm{s}$). When capacity constraints are added, in every feasible profile $\bm{s}$ it holds that $x_e(\bm{s}) \leq c_{max}$. Therefore, we can bound the cost of every SE by $\Phi(\bm{s}^*) \leq \sum_{e \in E}{p_e \cdot H_{c_{max}}} = H_{c_{max}} \cdot cost(\bm{s}^*)$.
\end{remark}

The following theorem extends the class of networks for which the SPoA is bounded by $H_n$.

\begin{theorem}\label{series-comp-spoa}
For every symmetric CFCSC game played on a network that is a series composition of EP networks, it holds that $SPoA \leq H_n$ (if a SE exists).
\end{theorem}
\begin{proof}
Let $G=G_1 \rightarrow G_2 \rightarrow \ldots \rightarrow G_k$ be a network. It is sufficient to prove that the SPoA of a game played on $G$ is upper bounded by the maximal SPoA of a game played on one of the subnetworks $G_1,\ldots,G_k$. Denote by $cost_{G_i}(\bm{s})$ the social cost of the strategy profile $\bm{s}$ in a game played only on the subnetwork $G_i$, and let $\bm{s}^*$ be an optimal profile. For every strategy profile $\bm{s}$, it holds that
\[
cost_G(\bm{s})=cost_{G_1}(\bm{s}) + \ldots + cost_{G_k}(\bm{s}).
\]
If $\bm{s}$ is a SE in $G$, the induced strategy profile in $G_i$ is also a SE. We get that
\begin{eqnarray*}
SPoA(G) & = & \max_{\bm{s}\text{ is SE}}{\frac{cost_G(\bm{s})}{cost_G(\bm{s}^*)}} \\
& = & \max_{\bm{s}\text{ is SE}}{\frac{cost_{G_1}(\bm{s}) + \ldots + cost_{G_k}(\bm{s})}{cost_{G_1}(\bm{s}^*) + \ldots + cost_{G_k}(\bm{s}^*)}} \\
& \leq & \max_{\bm{s}\text{ is SE}} \max_i{\frac{cost_{G_i}(\bm{s})}{cost_{G_i}(\bm{s}^*)}} \\
& \leq & \max_i{SPoA(G_i)}.
\end{eqnarray*}
\end{proof}

Note that SPP networks are a special case of the class specified in the last theorem, and therefore the upper bound of $H_n$ applies to SPP networks as well.

\subsection{SP and General Networks}
For SP networks, it is established in \cite{capacitated} that $PoA \leq n$, which directly implies that $SPoA \leq n$.
We provide an example showing that this bound is tight.

\begin{theorem}\label{thm:sp-network-spoa-n}
For every $\epsilon>0$, there exists an SP network $G$ and a CFCSC game played on $G$ such that $SPoA \geq \frac{n}{1+\epsilon}$.
\end{theorem}

\begin{figure}
\centering
\includegraphics{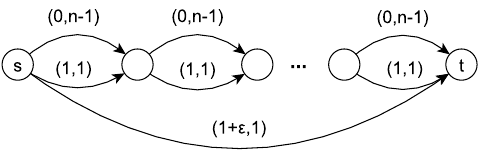}
\caption{\label{fig:SP-SPoAn}A game played on an SP network with $SPoA \geq \frac{n}{1+\epsilon}$}
\end{figure}

\begin{proof}
Consider the game played by $n$ agents on the graph depicted in Figure \ref{fig:SP-SPoAn}.
The strategy profile in which each agent uses a different edge of cost $1$ and $n-1$ additional edges of cost $0$ is a SE: In this profile, each agent pays $1$, and the only way to reduce this cost is by using a path of edges that cost $0$. However, the edges that cost $0$ are already used to their maximum capacity, and any deviation that allows an agent to use a path that costs $0$ will increase the cost of another agent.

The optimal profile in this game is the profile where one agent uses the lower edge of cost $1+\epsilon$ and $n-1$ agents use the upper edges of cost $0$. It follows that $SPoA \geq \frac{n}{1 + \epsilon}$, as stated.
\end{proof}

For general networks, the SPoA can be arbitrarily high, even with only two agents.

\begin{theorem}\label{example-unbounded-spoa}
For every real number $R$, there exists a CFCSC game with two agents in which the SPoA is greater than $R$.
\end{theorem}

\begin{figure}
\centering
\includegraphics{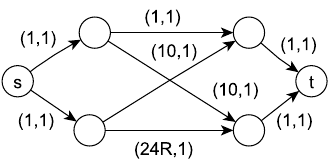}
\caption{\label{fig:unbounded-SPoA}A game with unbounded SPoA}
\end{figure}

\begin{proof}
Consider the game played by two agents on the graph presented in Figure \ref{fig:unbounded-SPoA}. The strategy profile in which no agent uses the inner edges is a SE: The agent that uses the upper path incurs the lowest possible cost and will not participate in any deviation, while the other agent uses the only path that is still available. However, the optimal strategy profile will avoid the edge of cost $24R$. Thus, $SPoA \geq \frac{24R+5}{24} > R$.
\end{proof}

\subsection{Homogeneous Capacities}
We consider a more restricted form of capacity constraints, in which all the edges have the same capacity. In this special case, we show that every SE in games played on EP networks is optimal.

\begin{theorem}
For every feasible symmetric CFCSC game with homogeneous capacities played on an EP network, it holds that $SPoA = 1$.
\end{theorem}
\begin{proof}
Let $\Delta$ be a feasible symmetric game played on an EP network $G$, and denote by $c$ the capacity of the edges. We show that every SE in $\Delta$ is optimal. We first show that there is a game that is played on a network of parallel edges, such that every SE in that game is optimal if and only if every SE in $\Delta$ is optimal. Let $G_1, \ldots, G_k$ be EP networks such that $G = G_1 \parallel G_2 \parallel \ldots \parallel G_k$, and none of $G_1,\ldots,G_k$ can be represented as a parallel composition of two networks. Since $G_i$ is an EP network, and cannot be decomposed into two parallel subnetworks, $G_i$ must be a single edge or an EP network that is extended by a single edge. It follows that in every strategy profile, no more than $c$ agents can use edges that belong to $G_i$. If $G_i$ is a single edge, its capacity is $c$, and if $G_i$ is an EP network extended by a single edge $e$, every agent that uses edges from $G_i$ must also use $e$, which has a capacity of $c$ agents.

Next, we claim that in every SE, all the agents that use edges from $G_i$ use the same path, which is the lowest-cost path in $G_i$ (if there is more than one path with minimal cost, one of these paths will be shared by all the agents). If there are agents who use a path of higher cost in $G_i$, they can deviate to the lowest-cost path and reduce their costs. This deviation satisfies the capacity constraints, as at most $c$ agents use edges from $G_i$, and the capacity of all the edges in the path is $c$.

Using the last observation, we can look at a different game, that is played on a network of parallel edges, and that every SE in $\Delta$ corresponds to a SE in the new game. For each of $G_1, \ldots, G_k$, we compute the minimum cost of an $s$-$t$ path that goes through that subnetwork, and replace the subnetwork with an edge that has that cost and a capacity of $c$ agents. Given a SE in the original game, we construct a strategy profile for the new game, such that the agents who use the minimal cost path in $G_i$ will use the edge that replaced $G_i$ in the new network. The costs of these two strategy profiles are the same for each agent.

It is only left to show that when all the capacities are equal, every SE in a game played on a network of parallel edges is optimal. Let $\bm{s}$ be a SE in a game played on a network of parallel edges, where the capacity of each edge is $c$. Denote by $e$ the edge that has the maximal cost of all the edges that are used in $\bm{s}$. Assume (by contradiction) that there is another edge $e'$, such that $p_{e'} < p_e$ ($e'$ is cheaper) and $x_{e'}(\bm{s}) < c$ ($e'$ is available). The number of agents who can deviate from $e$ to $e'$ is at most $\min \{x_e(\bm{s}), c - x_{e'}(\bm{s})\}$. If this number of agents deviate to $e'$, their cost will reduce from $\frac{p_e}{x_e(\bm{s})}$ to at most $\frac{p_{e'}}{x_e(\bm{s})}$. This contradicts our assumption that $\bm{s}$ is a SE. Therefore, every SE uses $\left \lceil \frac{n}{c} \right \rceil$ edges of minimal total cost, as does the optimal solution. We conclude that $SPoA = 1$.
\end{proof}

By applying the same arguments used in the proof of Theorem \ref{series-comp-spoa}, the last result can be extended to every network that is a series composition of EP networks, particularly, SPP networks.

\begin{corollary}
For every feasible symmetric CFCSC game with homogeneous capacities, that is played on a network that is a series composition of EP networks, it holds that $SPoA = 1$.
\end{corollary}

When considering more general networks, limiting the capacities to be homogeneous does not improve the SPoA. In the example presented in the proof of Theorem \ref{thm:sp-network-spoa-n}, which shows that the SPoA in SP networks can be arbitrarily close to $n$, we can replace each edge that costs $0$ with $n-1$ parallel edges that cost $0$ and have unit capacity. This way all the edges have unit capacity, and the theorem still holds for homogeneous capacities. In the example that shows that the SPoA can be unbounded (in non-SP networks), all the edges already have unit capacities.

\section{Strong Price of Stability}
In some cases, the best SE may be of interest as well, for example, if there exists a central entity that can coordinate the agents around an initial equilibrium. Clearly, it always holds that $PoS \leq SPoS \leq SPoA$ (since every SE is a NE), so the upper bounds on the SPoA apply to the SPoS as well.
Interestingly, the upper bounds on the SPoS in all the classes of network topologies considered here match the upper bounds on the SPoA. This is in stark contrast to previous settings, which exhibited large gaps between worst-case and best-case equilibria.
In what follows, we show that all the upper bounds on the SPoA established in the previous section are tight with respect to the best SE.

For EP and SPP networks, the bound of $H_n$ is tight with respect to the best SE due to an example provided in \cite{capacitated}, in which $PoS = \frac{H_n}{1+\epsilon}$ for every $\epsilon > 0$. The example appears here in the proof of Theorem~\ref{thm:ep-spoa} (see Figure~\ref{fig:PoS-Hn-tightness}).

For SP networks, we have shown that $SPoA \leq n$, and the following theorem shows a game for which $SPoS = \Omega(n)$.

\begin{theorem}\label{thm:SPoS-omega-n}
There exists a symmetric CFCSC game played on an SP network, in which the SPoS is at least $\Omega(n)$.
\end{theorem}

\begin{proof}
Consider the game of $n$ agents depicted in Figure \ref{fig:SPoS_n_2}. We divide the network into two parts. The first part, which we refer to as the upper part of the network, contains the paths that go from $s$ to $t$ through the nodes $v_1, \ldots, v_{n-2}$. In the upper subgraph, there are $n-1$ edges that cost $0$ and $n-1$ edges that cost $1$. The second part, which we refer to as the lower part of the network, contains the edge connecting $s$ and $t$ directly, and the two paths that go from $s$ to $t$ through node $u$. The upper part of the network can be used by at most $n-1$ agents, so one of the agents must use the lower part of the network.

\begin{figure}
\centering
\includegraphics{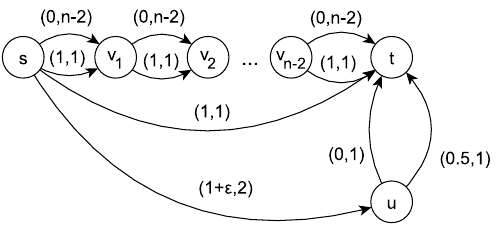}
\caption{\label{fig:SPoS_n_2}A game played on an SP network with $SPoS=\Theta(n)$}
\end{figure}

We show that the only SE in this game are profiles in which each agent pays $1$: $n-1$ agents use the upper part of the network. Each one of them uses a different edge of cost $1$ and $n-2$ edges of cost $0$. The last agent uses the edge that goes directly from $s$ to $t$ and costs $1$. We show that these profiles are SE by considering the possibly profitable deviations.
There are two possible paths that cost less than $1$. The first such path requires the agent to share the edge $(s,u)$ with another agent, and then use the left edge $(u,t)$. However, a deviation of this kind would require another agent to use the path $(s,u,t)$ that goes through the right edge $(u,t)$. That path costs $1 + \frac{\epsilon}{2}$, and the agent using it could not benefit from the deviation.
The other path that costs less than $1$ is the path that goes only through the edges that cost $0$. However, these edges are used to their maximum capacity, and any coalitional deviation that allows an agent to use a path of cost $0$ will not be profitable for one of the participating agents.
Since there are no profitable deviations, the profiles are SE, and their social cost is $n$.

Now, we show that all other profiles are not SE. In every other profile in which only one agent uses the lower subgraph, one of the agents who use the upper part of the network pays at least $2$. This agent can deviate to one of the lower paths and reduce her cost. Therefore, these profiles are not SE. Every profile in which there are two agents using the lower subgraph cannot be a SE either. In such profiles, the agents that use the upper part of the network will only use the edges that cost 0 and will not participate in any deviation. The remaining two agents always have a profitable deviation: In every NE, one of them will use the edge $(s,t)$ that costs 1. However, this agent's cost can be reduced by sharing the edge $(s,u)$ with the other agent, so that one of them pays $\frac{1+\epsilon}{2}$ and the other pays $1+\frac{\epsilon}{2}$ instead of $1+\epsilon$. It remains to prove that the profiles in which three agents use the lower subgraph are not SE, but these are not even NE, as one of the agents can benefit by deviating to a path consisting of edges of cost 0.

So far we have shown that the social cost of every SE is $n$. The optimal profile uses the edges that cost $0$, and the two paths that go through node $u$. The social cost of this profile is $1.5+\epsilon$, hence it holds that $SPoS = \frac{n}{1.5+\epsilon} = \Omega(n)$.
\end{proof}

Next, we show that in the case of non-SP networks, the SPoS can be unbounded (as shown for the SPoA).

\begin{theorem}\label{thm:SPoS-unbounded}
There exists a symmetric CFCSC game with two agents in which the SPoS can be arbitrarily high.
\end{theorem}

\begin{figure}
\centering
\includegraphics{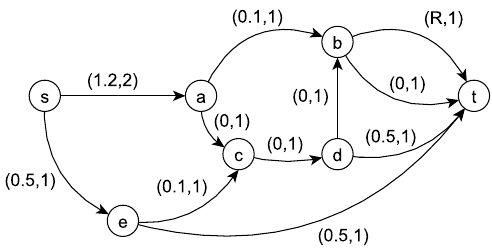}
\caption{\label{fig:unbounded-SPoS2}A game with unbounded SPoS}
\end{figure}

\begin{proof}
Consider the game of two agents presented in Figure \ref{fig:unbounded-SPoS2}. The following strategy profile is a SE: One agent uses the path $(s,e,c,d,b,t)$ and incurs the cost $0.6$, and the other agent uses the path $(s,a,b,t)$ and incurs the cost $R+1.3$. This is a SE since the first agent incurs the minimal possible cost, and the other agent uses the only available path. We claim that this is the only SE.

In every NE, it is impossible that two agents will use the edge $(s,a)$, as one of them incurs a cost of at least $1.1$, and therefore may deviate to the path $(s,e,t)$. Hence, one agent must use the edge $(s,a)$ and the other must use the edge $(s,e)$.

If one agent uses the path $(s,e,t)$, using the path $(s,a,c,d,b,t)$ minimizes the cost of the other agent. This strategy profile is not a SE, as they can jointly deviate to the paths $(s,a,b,t)$ and $(s,a,c,d,t)$, respectively, and both agents will benefit from the deviation.

We obtain that in every SE one agent must use the path $(s,e,c,d)$. This agent will not use the path $(s,e,c,d,t)$, as he can use the better path $(s,e,t)$. The agent will not use the edge that costs $R$ for the same reason.

We conclude that the only SE is the one presented above. This SE uses the edge that costs $R$. An optimal strategy profile will avoid that edge, as $R$ can be arbitrarily high. Therefore, the SPoS is unbounded, as stated.
\end{proof}

\section{Extensions}
\subsection{Asymmetric Games}
A natural extension of our model is the case of asymmetric games, where different agents are associated with different source and sink nodes.
In this section we remark on the existence of SE and their quality in asymmetric games. Note that we still consider games that are played on networks with designated source and sink nodes (as defined earlier), but the source and sink nodes of network do not necessarily serve as the source and sink of the individual agents.

Regarding the existence of SE, Epstein et al.\ \cite{CostSharingSE} provided an example of an asymmetric game played on an SPP network that does not admit a SE. The following two theorems provide a characterization of the networks that admit a SE: one theorem considers games with a single source and the other considers games with multiple source nodes. The proofs are deferred to the appendix.

\begin{theorem}\label{thm:asymmetric-se-existence}
In single-source games, a symmetric network $G$ admits a SE if and only if it is SPP.
\end{theorem}

\begin{theorem}\label{thm:multiple-source-existence}
In games with multiple source and sink nodes, a symmetric network $G$ admits a SE if and only if there are networks $G_1,\ldots,G_k$ such that $G=G_1 \rightarrow G_2 \rightarrow \ldots \rightarrow G_k$, and that one of $G_1,\ldots,G_k$ is a network of parallel paths and the rest are networks of parallel edges.
\end{theorem}

Next we study the SPoA in asymmetric games.
Unlike the symmetric case, which exhibited reasonable bounds for some families of topologies, in the asymmetric case, even if all agents share the same source (but not the same sink) and play a game on a simple EP network, the SPoA may be unbounded.

\noindent{\em Example.}
Consider the game played by two agents that is presented in Figure \ref{fig:asymmetric-SP-unbounded-SPoA2}. The source node of both agents is $s$, and the sink nodes are $t_1$ and $t_2$, respectively. The strategy profile in which agent 1 uses the edge of cost $R$ and agent 2 uses the path of cost $0$ is a SE. In the optimal profile, agent 1 uses the edge of cost $0$, and agent 2 uses the edge of cost $1$. Therefore, in this example $SPoA = R$, which can be arbitrarily high.

\begin{figure}
\begin{center}
\includegraphics{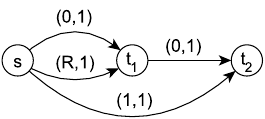}
\caption{\label{fig:asymmetric-SP-unbounded-SPoA2}An asymmetric game with unbounded SPoA}
\end{center}
\end{figure}

However, for the class of SPP networks, the SPoA is upper bounded by $H_n$, even for asymmetric games.

\begin{theorem}\label{thm:asymmetric-spp-spoa}
For every CFCSC game (either symmetric or asymmetric) played on an SPP network, it holds that $SPoA \leq H_n$ (if a SE exists).
\end{theorem}

The proof is deferred to the appendix.

\subsection{Undirected Graphs}
Up until now we considered the case of cost-sharing games played on a directed graphs.
In this section we show that our results (for symmetric games) extend to undirected graphs.

For directed graphs we have shown that a network admits a SE if and only if it is an SPP network.
Theorem~\ref{thm:spp-se-existence} asserts that every SPP network admits a SE.
To prove existence we use a greedy algorithm for assigning edges to agents, such that in every coalition, there must exist an agent that cannot reduce his cost in any subnetwork of parallel edges.
The same algorithm can be used to compute a SE in the undirected case (due to exactly the same arguments).
In the converse direction, we presented two examples of games played on directed graphs that admit no SE.
These games admit no SE also in the case where the underlying graph is undirected.
To complete the argument we note that the lemma that is used to extend these forbidden networks to networks that embed them applies to undirected graphs following an identical analysis.\footnote{In the undirected case, we use the characterization of SP networks shown by Milchtaich~\cite{milchtaich2006network}.}

For bounding the SPoA and the SPoS in EP and SP networks, we use the following claim.

\begin{claim}
Let $\Delta$ be a symmetric CFCSC game played on a directed SP network $G$, and let $\Delta'$ be the same game played on the same network but with undirected edges. If $\bm{s}$ is a SE in $\Delta'$, then $\Delta$ has a SE with the same social cost as in $\bm{s}$.
\end{claim}

\begin{proof}
Let $\bm{s}$ be a SE in $\Delta'$.
We first claim that we can assume, without loss of generality, that $\bm{s}$ uses only acyclic paths.
To see this note that the paths in $\bm{s}$ cannot contain cycles that cost more than $0$, as if one of them did, the agent using that path could reduce her cost by removing the cycle.
In addition, it can be assumed that the paths in $\bm{s}$ do not contain cycles that cost $0$, as removing these cycles does not affect the agent costs.
To conclude the proof, we use the following lemma, established by Milchtaich \cite{milchtaich2006network}.

\begin{lemma}
\cite{milchtaich2006network} Let $G$ be an undirected SP network, and let $u,v$ be two distinct vertices in $G$. If there exists an acyclic $s$-$t$ path in which $u$ precedes $v$, then $u$ precedes $v$ in every acyclic $s$-$t$ path that contains both vertices.
\end{lemma}

By the above lemma, it follows that every edge used in the strategy profile $\bm{s}$ (in the game $\Delta'$) is used in the direction determined by $G$. Each edge $(u,v)$ in the directed SP network $G$ can be used as part of an acyclic path in $G$, and the lemma guarantees that if $u,v$ are both reached in a path in $\bm{s}$, $u$ precedes $v$. Therefore, $\bm{s}$ is a valid strategy profile in $\Delta$. Clearly, any deviation that is feasible in $\Delta$ is also feasible in $\Delta'$; therefore the profile $\bm{s}$ is a SE in $\Delta$ as well.
\end{proof}

The above claim implies that the set of SE in the undirected case is a subset of the set of SE in the directed case (up to cycles of cost 0). Therefore, upper bounds on the SPoA in games played on directed SP networks carry over to the undirected case.
In addition, the examples that are used to show that these upper bounds are tight (or that the SPoA can be unbounded in general networks) hold when the edges are undirected as well.
In the case of games with homogeneous capacities that are played on EP and SPP networks, we have shown that every SE is optimal, and this result carries over to the undirected case (by the above claim).

The upper bounds on the SPoS are derived from the fact that $SPoS \leq SPoA$, and we have shown that these bounds are asymptotically tight. The examples that are used to show that these upper bounds are tight still hold when the edges are undirected (the analysis of the game with unbounded SPoS appears in the appendix). This completes the extension of our results for symmetric games to the case of undirected graphs.

\section*{Acknowledgments}
Idan Nurick was involved in a preliminary version of this work. We would like to thank Nick Gravin for helpful discussions.

\bibliography{se_capacitated_teac}

\begin{thebibliography}{10}

\bibitem{Albers}
Susanne Albers.
\newblock On the value of coordination in network design.
\newblock {\em SIAM J. Comput.}, 38(6):2273--2302, 2009.

\bibitem{StrongPriceOfAnarchy}
Nir Andelman, Michal Feldman, and Yishay Mansour.
\newblock Strong price of anarchy.
\newblock {\em Games and Economic Behavior}, 65(2):289 -- 317, 2009.

\bibitem{FairCostSharing}
Elliot Anshelevich, Anirban Dasgupta, Jon Kleinberg, {\'E}va Tardos, Tom
  Wexler, and Tim Roughgarden.
\newblock The price of stability for network design with fair cost allocation.
\newblock {\em SIAM Journal on Computing}, 38(4):1602--1623, 2008.

\bibitem{ConnectionGame}
Elliot Anshelevich, Anirban Dasgupta, {\'E}va Tardos, and Tom Wexler.
\newblock Near-optimal network design with selfish agents.
\newblock {\em Theory of Computing}, 4(4):77--109, 2008.

\bibitem{Aumann}
Robert~J. Aumann.
\newblock Acceptable points in general cooperative n-person games.
\newblock In A.~W. Tucker and R.~D. Luce, editors, {\em Contributions to the
  Theory of Games IV, Annals of Mathematical Study 40}, pages 287--324.
  Princeton University Press, 1959.

\bibitem{PoSUndirectedTCS}
Vittorio Bil{\`o}, Ioannis Caragiannis, Angelo Fanelli, and Gianpiero Monaco.
\newblock Improved lower bounds on the price of stability of undirected network
  design games.
\newblock {\em Theory of Computing Systems}, 52(4):668--686, 2013.

\bibitem{PoSUndirectedFOCS}
Vittorio Bil{\`{o}}, Michele Flammini, and Luca Moscardelli.
\newblock The price of stability for undirected broadcast network design with
  fair cost allocation is constant.
\newblock In {\em Proceedings of the 54th Annual IEEE Symposium on Foundations
  of Computer Science}, FOCS'13, pages 638--647, 2013.

\bibitem{StrongPareto}
Steve Chien and Alistair Sinclair.
\newblock Strong and pareto price of anarchy in congestion games.
\newblock In {\em Proceedings of the 36th International Colloquium on Automata,
  Languages and Programming: Part I}, ICALP'09, pages 279--291, 2009.

\bibitem{CostSharingSE}
Amir Epstein, Michal Feldman, and Yishay Mansour.
\newblock Strong equilibrium in cost sharing connection games.
\newblock {\em Games and Economic Behavior}, 67(1):51--68, 2009.

\bibitem{furtherCapacitated}
Thomas Erlebach and Matthew Radoja.
\newblock Further results on capacitated network design games.
\newblock In {\em Proceedings of the 8th International Symposium on Algorithmic
  Game Theory}, SAGT'15, pages 57--68, 2015.

\bibitem{capacitated}
Michal Feldman and Tom Ron.
\newblock Capacitated network design games.
\newblock {\em Theory of Computing Systems}, pages 1--22, 2014.

\bibitem{LoadBalancing}
Amos Fiat, Haim Kaplan, Meital Levy, and Svetlana Olonetsky.
\newblock Strong price of anarchy for machine load balancing.
\newblock In {\em Proceedings of the 34th International Conference on Automata,
  Languages and Programming}, ICALP'07, pages 583--594, 2007.

\bibitem{CongestionStrongPotential}
Ron Holzman and Nissan Law-Yone.
\newblock Strong equilibrium in congestion games.
\newblock {\em Games and Economic Behavior}, 21(1–2):85 -- 101, 1997.

\bibitem{NonDecCongestionSE}
Ron Holzman and Nissan Law-Yone.
\newblock Network structure and strong equilibrium in route selection games.
\newblock {\em Mathematical Social Sciences}, 46(2):193 -- 205, 2003.

\bibitem{CongestionSE}
Ron Holzman and Dov Monderer.
\newblock Strong equilibrium in network congestion games: increasing versus
  decreasing costs.
\newblock {\em International Journal of Game Theory}, 44(3):647--666, 2015.

\bibitem{PriceOfAnarchy}
Elias Koutsoupias and Christos Papadimitriou.
\newblock Worst-case equilibria.
\newblock {\em Computer Science Review}, 3(2):65 -- 69, 2009.

\bibitem{LiMulticast}
Jian Li.
\newblock An o(log(n)/log(log(n))) upper bound on the price of stability for
  undirected shapley network design games.
\newblock {\em Information Processing Letters}, 109(15):876 -- 878, 2009.

\bibitem{milchtaich2006network}
Igal Milchtaich.
\newblock Network topology and the efficiency of equilibrium.
\newblock {\em Games and Economic Behavior}, 57(2):321--346, 2006.

\bibitem{PotentialGames}
Dov Monderer and Lloyd~S. Shapley.
\newblock Potential games.
\newblock {\em Games and Economic Behavior}, 14(1):124 -- 143, 1996.

\bibitem{CongestionGamesRosenthal}
Robert~W. Rosenthal.
\newblock A class of games possessing pure-strategy nash equilibria.
\newblock {\em International Journal of Game Theory}, 2(1):65--67, 1973.

\end{thebibliography}
\bibliographystyle{plain}

\appendix

\section{Omitted Proofs\label{appendix-proofs}}
\begin{proof}[Proof of Theorem~\ref{thm:asymmetric-se-existence}]
One direction has already been proven in Section~\ref{sec:existence}: For every non-SPP network, we have shown that there is a symmetric CFCSC game played on it that does not admit a SE. We now show that every single source game played on an SPP network admits a SE.

Let $\Delta$ be a single source game played on an SPP network $G$. We denote the common source node by $s$ and the sink node of agent $j$ by $t_j$. By the definition of SPP networks, $G$ is series composition of networks of parallel paths. Let $G_1, \ldots, G_k$ be networks of parallel paths such that $G=G_1 \rightarrow G_2 \rightarrow \ldots \rightarrow G_k$. For now, we assume that $s$ is the source node of $G_1$.

Consider the sink node of agent $j$. That node must be either an inner node or a sink node of one of the subnetworks $G_1, \ldots, G_k$. Assume that $t_j$ is an inner node or the sink of subnetwork $G_i$. Now let us look at a path from $s$ to $t_j$. Since the graph is directed, any path from $s$ to $t_j$ does not use edges from subnetworks $G_{i+1}, \ldots, G_k$. In addition, any path from $s$ to $t_j$ can be divided into $i$ subsequent parts. The first $i-1$ parts are paths from the source to the sink of each of the subnetworks $G_1, \ldots, G_{i-1}$. The last part is a path from the source node of $G_i$ to $t_j$. If $t_j$ is an inner node of $G_i$, there is only one path from the source of $G_i$ to $t_j$.

We have seen that the sink node of an agent determines whether edges from a subnetwork $G_i$ can be part of her path. Without loss of generality, we assume that the agents are ordered in a way that agents who use more subnetworks of $G_1, \ldots, G_k$ come first. Formally, for every two agents $j_1, j_2$, such that $t_{j_1}$ is an inner node or the sink node of $G_{i_1}$ and $t_{j_2}$ is an inner node or the sink node of $G_{i_2}$, if $j_1 < j_2$, then: (i) $i_1 \geq i_2$, (ii) if $t_{j_1}$ is an inner node and $t_{j_2}$ is not an inner node, then $i_1 > i_2$.

Now we define a strategy profile $\bm{r}$ and show that it is a SE. The definition of the profile uses Algorithm~\ref{compute-se-alg} in a similar way to the proof of Theorem~\ref{thm:spp-se-existence}. The profile $\bm{r}$ is defined in two steps. First, if agent $j$'s sink node is an inner node of subnetwork $G_i$, every path from $s$ to $t_j$ uses the same edges in $G_i$, so we assign these edges to the agent. This is only a part of agent $j$'s path, and the rest of her path is determined in the next step. In the next step, we iterate through the subnetworks $G_1, \ldots, G_k$. When subnetwork $G_i$ is considered, we determine how many agents need to be assigned a path from the source node to the sink node of $G_i$, and compute these paths using Algorithm~\ref{compute-se-alg}. The algorithm is stated for edges, but as in the proof of Theorem~\ref{thm:spp-se-existence}, each path can be replaced by an equivalent edge that sums the costs of the edges in the path and has the minimal available capacity of these edges. Then, we assign each of the paths returned by the algorithm to the agents according to their order: The first agent $j_1$ who needs to be assigned a path in $G_i$ gets the lowest-cost path returned by the algorithm. Then, the second agent $j_2 > j_1$ who needs a path in $G_i$ is assigned the second lowest-cost path, and we continue the same way until all the paths returned by the algorithm are assigned. This completes the definition of $\bm{r}$.

We claim that the strategy profile $\bm{r}$ is a SE. Assume by contradiction that there is a coalition $C$ for which deviating to another profile is profitable, and let $j$ be the minimal index of an agent in $C$. Agent $j$ must reduce the cost that she pays for the edges in at least one of the subnetworks $G_1, \ldots, G_k$.

For each $G_i$, agent $j$ can reduce her cost in $G_i$ by moving to another path, or by sharing the cost of her path with another agent that deviates to that path. As in the proof of Theorem~\ref{thm:spp-se-existence}, agent $j$ cannot use another path in a profitable deviation, as if there were such a path, the algorithm would have assigned that path to $j$ earlier. So after the deviation, agent $j$ must share her path in $G_i$ with more agents compared to the original profile $\bm{r}$. There are two cases:

\begin{enumerate}
\item Agent $j$'s path reaches the sink node of $G_i$. In that case, either agent $j$'s path in $G_i$ has an edge $e$ that is used by $c_e$ agents, or the path is used by agents $j, \ldots, n$. If the path contains an edge $e$ that is used by $c_e$ agents in $\bm{r}$, the path cannot be used by more than $c_e$ agents after the deviation. If the path is already used by agents $j, \ldots, n$, no other agents can join the path, since $j = \min{C}$. In both cases, the cost of agent $j$ in $G_i$ does not decrease.
\item Agent $j$'s sink node is an inner node of $G_i$. In that case, the only agents who can join agent $j$'s path in $G_i$ are agents who reach the sink node of $G_i$. However, the agents are ordered so that all the agents who reach the sink node of $G_i$ are indexed lower than $j$. These agents do not participate in the deviation, and we conclude that agent $j$ cannot reduce the cost of her path in $G_i$.
\end{enumerate}

We have shown that the cost of agent $j$ cannot decrease after the deviation, which contradicts the assumption that the deviation is profitable. Therefore, $\bm{r}$ is a SE.

In the proof, we assumed that the source node $s$ is the source node of $G_1$. If $s$ is the source node of another subnetwork $G_i$, the subnetworks $G_1, \ldots, G_{i-1}$ are never used by any agent and can be ignored. If $s$ is an inner node of subnetwork $G_i$, all the possible paths of all agents use the same edges in $G_i$, and no agent can reduce the cost of her path in $G_i$.
\end{proof}

\begin{proof}[Proof of Theorem~\ref{thm:multiple-source-existence}]
We first prove that the networks described in the theorem admit a SE. The proof is very similar to that of Theorem~\ref{thm:asymmetric-se-existence}. Assume that $G=G_1 \rightarrow G_2 \rightarrow \ldots \rightarrow G_k$, where $G_l$ is a network of parallel paths and $G_1,\ldots,G_{l-1},G_{l+1},\ldots,G_k$ are networks of parallel edges. Let $\Delta$ be a multi-source game played on $G$.

As in the proof of Theorem~\ref{thm:asymmetric-se-existence}, the source and sink nodes of each agent define which of the subnetworks $G_1,\ldots,G_k$ are used as part of her path. If the source or sink node of an agent is an inner node of $G_l$, there is only one path in $G_l$ that reaches that node and the agent must use it.

In this proof, we assume that the agents are ordered in a way that agents whose source or sink node is an inner node of $G_l$ come after all other agents. Formally, if $j_1$ is an agent that both her source and sink nodes are not inner nodes of $G_l$, and $j_2$ is an agent whose source or sink node is an inner node of $G_l$, then $j_1 < j_2$.

We now define a strategy profile $\bm{r}$ and show that it is a SE. The profile $\bm{r}$ is defined in two steps. First, if agent $j$'s source or sink node is an inner node of subnetwork $G_l$, every possible path for that agent uses the same edges in $G_l$, so we assign these edges to the agent. This is only a part of agent $j$'s path, and the rest of her path is determined in the next step. In the next step, we iterate through the subnetworks $G_1, \ldots, G_k$ (including $G_l$). When subnetwork $G_i$ is considered, we determine how many agents need to be assigned a path from the source node to the sink node of $G_i$, and compute these paths using Algorithm~\ref{compute-se-alg}. Then, we assign each of the paths returned by the algorithm to the agents according to their order: The first agent $j_1$ who needs to be assigned a path in $G_i$ gets the lowest-cost path returned by the algorithm. Then, the second agent $j_2 > j_1$ who needs a path in $G_i$ is assigned the second lowest-cost path, and so on. This completes the definition of $\bm{r}$.

We claim that the strategy profile $\bm{r}$ is a SE. Assume by contradiction that there is a coalition $C$ that benefits from deviating to another profile, and let $j$ be the minimal index of an agent in $C$.

Agent $j$ can reduce her cost in $G_i$ by moving to another path, or by sharing the cost of her path with another agent that deviates to that path. As in the proof of Theorem~\ref{thm:spp-se-existence}, agent $j$ cannot use another path in a profitable deviation, as if there were such a path, the algorithm would have assigned that path to $j$ earlier. So after the deviation, agent $j$ must share her path in $G_i$ with more agents compared to the original profile $\bm{r}$. There are two cases:

\begin{enumerate}
\item Agent $j$'s path contains both the source and sink nodes of $G_i$. In that case, either agent $j$'s path in $G_i$ has an edge $e$ that is used by $c_e$ agents, or the path is used by agents $j, \ldots, n$. If the path contains an edge $e$ that is used by $c_e$ agents in $\bm{r}$, the path cannot be used by more than $c_e$ agents after the deviation. If the path is already used by agents $j, \ldots, n$, no other agents can join the path, since $j = \min{C}$. In both cases, the cost of agent $j$ in $G_i$ does not decrease.
\item $i = l$ and agent $j$'s source or sink node is an inner node of $G_l$. In that case, the only agents who can join agent $j$'s path in $G_l$ are agents whose path contains both the source and sink nodes of $G_l$. However, the agents are ordered so that these agents are indexed lower than $j$. These agents do not participate in the deviation, and we conclude that agent $j$ cannot reduce the cost of her path in $G_l$.
\end{enumerate}

We have shown that the cost of agent $j$ cannot decrease after the deviation, which contradicts the assumption that the deviation is profitable. Therefore, $\bm{r}$ is a SE.

\begin{figure}
\centering
\includegraphics{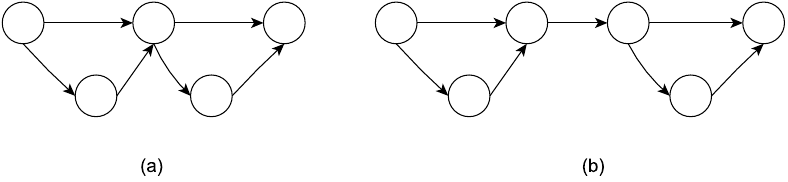}
\caption{\label{fig:SPP_No_SE}SPP networks that do not admit a SE in games with multiple source and sink nodes}
\end{figure}

The converse direction relies on previous proofs. For every non-SPP network, we have already shown that there is a symmetric CFCSC game played on it that has no SE. Epstein et al.\ \cite{CostSharingSE} have provided an example of a game with multiple source and sink nodes that is played on the network depicted in Figure~\ref{fig:SPP_No_SE}(a) and does not admit a SE. An equivalent game can be played on the network presented in Figure~\ref{fig:SPP_No_SE}(b), by setting the cost of the additional edge to be $0$. It is easy to show that one of the networks of Figure~\ref{fig:SPP_No_SE} is embedded in every SPP network $G=G_1 \rightarrow G_2 \rightarrow \ldots \rightarrow G_k$ that at least two of its subnetworks $G_1, \ldots, G_k$ are not networks of parallel edges. By the arguments used in the proof of Theorem~\ref{thm:no-se-on-non-spp}, it follows that these networks do not admit a SE.
\end{proof}

\begin{proof}[Proof of Theorem~\ref{thm:asymmetric-spp-spoa}]
Let $\Delta$ be a game played on an SPP network $G$ such that $G=G_1 \rightarrow G_2 \rightarrow \ldots \rightarrow G_k$, and let $\bm{s}$ be a SE in $\Delta$. As in the proof of Theorem~\ref{thm:ep-spoa}, it is sufficient to show that there is an optimal strategy profile $\bm{s}^*$, such that for every coalition of agents $C$, the profile $(\bm{s}_C,\bm{s}^*_{-C})$ is feasible.

Let $G_{OPT}$ be the subnetwork of $G$ that is used by an optimal solution. We define a strategy profile $\bm{s}^*$ that uses edges only from $G_{OPT}$. The source and sink nodes of each agent define which of the subnetworks $G_1,\ldots,G_k$ are used as part of her path. We iterate through $G_1, \ldots, G_k$, and for every agent who uses edges from $G_i$, we define the path in $G_i$ that this agent uses in $\bm{s}^*$.

First, for each of the parallel paths in $G_i$ that are used in both $\bm{s}$ and $G_{OPT}$, we assign each path to the same agents who use it in $\bm{s}$. This includes all the agents whose source or sink node is an inner node of $G_i$, who must use a specific path in $G_i$ in every strategy profile. This does not violate the capacity constraints due to the feasibility of $\bm{s}$. Then, for the other agents who use $G_i$ and still do not have a path, we assign them any available path in $G_i$ that is part of $G_{OPT}$. Since there is a feasible strategy profile that uses only edges from $G_{OPT}$, there are enough available paths. This completes the definition of $\bm{s}^*$.

We now prove that for every coalition of agents $C$, the profile $\bm{s}_{comb} = (\bm{s}_C,\bm{s}^*_{-C})$ is feasible. Let $e$ be an edge in $G_i$. Denote by $M$ the set of agents that use $e$ in $\bm{s}$, and let $M^*$ denote the set of agents that use $e$ in $\bm{s}^*$. The set of agents that use $e$ in $\bm{s}_{comb}$ is $(M \cap C) \cup (M^* \cap (N \backslash C))$, where $N$ is the set of all agents. If $e$ is used only in $\bm{s}$, it cannot exceed its capacity in $\bm{s}_{comb}$, since $M \cap C \subseteq M$, and the profile $\bm{s}$ is feasible. The same applies to edges that are used only in $\bm{s}^*$. It remains to consider edges that are used in both $\bm{s}$ and $\bm{s}^*$. In the definition of $\bm{s}^*$, every agent who uses the simple path in $G_i$ that contains $e$ in $\bm{s}$ is also assigned the same path in $\bm{s}^*$. It follows that $M \subseteq M^*$, and $(M \cap C) \cup (M^* \cap (N \backslash C)) \subseteq M^*$. Since $\bm{s}^*$ is a feasible solution, we conclude that $e$ does not exceed its capacity in $\bm{s}_{comb}$. Therefore, $\bm{s}_{comb}$ is feasible.
\end{proof}

\begin{proof}[Proof of Theorem~\ref{thm:SPoS-unbounded} for Undirected Graphs]
The extension of the theorem to undirected graphs has two parts. First, we have to show that the SE described in the proof is still a SE. Recall that in this profile, one agent uses the path $(s,e,c,d,b,t)$ and pays $0.6$, and the other agent uses the path $(s,a,b,t)$ and pays $R + 1.3$. Note that the lowest possible cost an agent can incur is still $0.6$, which means that the agent who pays $0.6$ will not participate in a deviation. The other agent cannot deviate to any path that excludes the edge of cost $R$. Hence, the profile is a SE.

Second, we have to show that when the edges are undirected, there are no new SE that exclude the edge $(b,t)$ that costs $R$. Note that any path in a SE does not contain a cycle that costs more than $0$. We analyze the strategy profiles that use edges in the opposite direction in order to look for new SE.

Using any of the edges that are adjacent to the source or the sink in the opposite direction results in a cycle of non-zero cost. Hence, this cannot occur in a SE.

Any path that uses the edge $(c,e)$ (in that direction) either contains the edge $(e,t)$ or the edge $(e,s)$. If an agent uses both $(c,e)$ and $(e,t)$, the edge $(s,e)$ cannot be used by the other agent, and deviating to the path $(s,e,t)$ is profitable. If the path of an agent contains the subpath $(c,e,s)$, the path must contain a cycle, which is impossible in a SE.

Any feasible profile in which an agent uses one of the edges $(b,a)$, $(c,a)$, or $(d,c)$ (in these directions) must contain a cycle that costs more than $0$.

Finally, the only feasible profile without cycles in which the edge $(b,d)$ is used (in that direction) contains the two paths $(s,a,b,d,t)$ and $(s,e,t)$. The costs of the paths are $1.8$ and $1$, respectively, and the agents can benefit by deviating to another profile, such as $(s,a,c,d,t)$ and $(s,a,b,t)$. We get that every SE must use the edge that costs $R$ (and $R$ can be arbitrarily high).
\end{proof}

\end{document}